\documentclass[a4paper]{article}

\usepackage{amsfonts}
\usepackage{amssymb}
\usepackage{amsmath}
\usepackage{amsthm}
\usepackage{amscd}
\usepackage{graphicx}
\usepackage{color}
\usepackage{bbm}
\usepackage{mathrsfs}
\usepackage{hyperref}

\theoremstyle{plain}

\newtheorem{corollary}{Corollary}[section]
\newtheorem{proposition}{Proposition}[section]

\theoremstyle{definition}
\newtheorem{definition}{Definition}[section]

\newtheorem{assumption}{Assumption}[section]

\theoremstyle{remark}

\usepackage{color}
\definecolor{r}{rgb}{1,0,0}

\newcommand*{\be}{\begin{equation}}
\newcommand*{\ee}{\end{equation}}
\newcommand*{\bea}{\begin{eqnarray}}
\newcommand*{\eea}{\end{eqnarray}}
\newcommand*{\bd}{\begin{displaymath}}
\newcommand*{\ed}{\end{displaymath}}
\newcommand*{\nn}{\nonumber\\}

\newcommand*{\mc}{\mathcal}
\newcommand*{\mf}{\mathfrak}

\newcommand*{\mb}{\mathbb}

\newcommand*{\id}{\textrm{id}}

\newcommand*{\1}{\mathbbm{1}}

\def\ba#1\ea{\begin{align}#1\end{align}}
\def\bas#1\eas{\begin{align*}#1\end{align*}}

\title{Spacetime granularity from finite-dimensionality\\of local observable algebras}

\author{{\bf Matti Raasakka}\vspace{5pt}\\\tt{matti.raasakka@aalto.fi}\vspace{5pt}\\\small Micro and Quantum Systems group\\\small Department of Electronics and Nanoengineering\\\small School of Electrical Engineering\\\small Aalto University}

\date{\today}

\begin{document}

\maketitle

\begin{abstract}
There are important indications that nature may be locally finite-dimensional, i.e., that any spatially bounded subsystem can be described by a finite-dimensional local observable algebra. Motivated by these ideas, we show that operational spacetime topology is described by an atomistic Boolean algebra if (i) local observable algebras are finite-dimensional factors, (ii) the intersection of two local algebras is also local, and (iii) the commutant of a local algebra is also local. Thus, in this case, spacetime has a point-free granular behavior at small scales.
\end{abstract}

\section{Introduction}
The aim of this paper is to show that, by rather general assumptions, local finite-dimensionality of physics leads to modifications of spacetime topology at small scales. By local finite-dimensionality we mean that physics in any bounded spacetime region $\mc{O}$ can described in terms of a finite-dimensional Hilbert space of states $\mc{H}_\mc{O} \cong \mb{C}^n$, $n\in\mb{N}$. Accordingly, also the local observable algebra $\mf{A}_\mc{O} \subset B(\mc{H}_\mc{O})$ associated to any bounded spacetime region is a finite-dimensional factor, and thus isomorphic to a full matrix algebra.

There are several reasons to suspect that nature should be fundamentally locally finite-dimen-sional, even though quantum field theory (QFT) is not. (See, e.g., \cite{Bao17} for a recent argument.) The high energy divergencies of QFT suggest that we should think of QFT as an effective theory to be replaced by some other model of physics in the deep UV. Gravity becomes, of course, relevant at the Planck scale, which will necessarily modify the theory. The most basic motivation for local finite-dimensionality is the belief that it should not be physically possible to store an infinite amount of information into an arbitrarily small spacetime region. This belief is backed up by Bekenstein's bound on the entropy of gravitational systems: The entanglement entropy of the QFT vacuum state restricted to a spatial subregion is always UV divergent, whereas according to Bekenstein's seminal work \cite{Bekenstein73,Bekenstein81} (and later works by others, e.g., \cite{Bousso99,Bousso00}) bounded gravitating systems should be able to carry only a finite amount of entropy. Thus, gravity should somehow regulate the UV behavior of quantum fields, perhaps through a UV cut-off at the Planck scale. Indeed, if the total energy in a spatial region of linear size $l$ exceeds the value $E_{BH} \sim l/G$, the region forms a black hole and thus cannot be observed from the outside. A cut-off to the total energy of bounded systems leads immediately to a locally finite-dimensional theory, because only a finite number of field modes can be excited in this case. If nature is fundamentally locally finite-dimensional, QFT must then arise as an infinite-dimensional approximation to the more accurate finite-but-extremely-high-dimensional model of physics in macroscopic spacetime regions. 

Our motivation to study locally finite-dimensional quantum physics comes more specifically from attemps to understand gravity as an effective phenomenon arising from the statistical properties of QFT. Jacobson and collaborators \cite{Jacobson95,Jacobson12,Jacobson15,Jacobson19} have shown that gravity may emerge from the entanglement first law for quantum field states if the theory has a physical UV cut-off at the Planck scale (and satisfies a number of other physically motivated assumptions), which leads to a finite entanglement entropy of the restricted vacuum state. The gravitational constant $G$ is then related to the entanglement entropy density $\eta$ on spatial 2-surfaces via $G=1/4\eta$. When the UV cut-off is removed, $\eta \rightarrow \infty$ and thus gravitational interactions vanish as $G \rightarrow 0$. A related argument for the necessity of local finite-dimensionality for the emergence of gravity stems from the fact that, due to the Hadamard condition \cite{Fewster13} (and generalizations thereof \cite{Radzikowski96}), finite-energy states in QFT have the same UV divergence structure as the vacuum state. Therefore, finite-energy perturbations of the vacuum cannot change the effective geometry associated with the area law, and gravitational effects \a la Jacobson cannot appear. Finite dimensionality of the local algebras allows physical perturbations to change the effective background geometry of the system associated with the area law. Of course, for any macroscopic region the dimensionality of the local algebra is extremely large, and therefore it takes a highly energetic perturbation to change the entanglement entropy of the system significantly, which may explain the weakness of gravity at macroscopic scales.

A common argument against finite-dimensionality of physics is the implied violation of Lorentz invariance, which might be carried over from the UV into the IR by perturbative corrections \cite{Mattingly05}. However, when we discuss local regions of spacetime (even in Minkowski spacetime), Lorentz transformations cannot be defined inside a local region, since such a region is never preserved under Lorentz transformations. Therefore, Lorentz invariance cannot be required inside a local region. Of course, the descriptions of local regions connected by global Lorentz transformations should still agree in the case of Minkowski spacetime. Despite local finite-dimensionality, the global algebra is still infinite-dimensional, and thus allows for global Lorentz symmetry. In this context it is relevant that a UV cut-off cannot be implemented in a Lorentz invariant manner, since the energy of any excitation can be arbitrarily increased by a Lorentz boost. On the other hand, e.g., the maximal spatial volume of a non-extendible spatial hypersurface inside the local region is invariant under diffeomorphisms. Perhaps the dimensionality of a local system could be related to the maximal spatial volume. Another possibility is the \emph{holographic principle} \cite{Susskind95, Bousso00b}, according to which the local dimensionality is related to the area of the spatial boundary of a region. We will tentatively assume the former option in the following, but our argument for spacetime granularity should be general enough to cover the latter one as well.

Of course, when gravity becomes relevant we should not expect to have global Lorentz symmetry in general, but (at most) local Lorentz covariance in agreement with the equivalence principle. In \cite{Raasakka17} we showed how local Lorenz covariance may appear in the locally finite-dimensional context as transformations between local thermal Hamiltonians: If the local algebras associated to minimal spatial regions are isomorphic to the observable algebra of a qubit (i.e., a 2-by-2 matrix algebra), then local thermal states on any two of these minimal local algebras can be transformed to each other via a unique $SL(2,\mb{C})$ transformation of the thermal Hamiltonian. In this way we can recover a Lorentz connection on the minimal local spacetime regions. In this paper, however, we will focus rather on the topological consequences of local finite-dimensionality. We will stay agnostic about the exact form of the local subalgebras, except for the assumption that they are finite-dimensional factors.

Our discussion in the rest of this paper will rely on the formalism of algebraic QFT, which is fundamentally based on the assignment of algebras of local operators to bounded spacetime regions. However, the existence of local observables in a quantum gravitational theory is highly doubtful \cite{Torre93, Marolf15,Giddings15,Donnelly16}\footnote{I thank Ted Jacobson for bringing \cite{Marolf15} to my attention, and an anonymous referee for pointing out \cite{Torre93}.}. For example, Giddings et al.\ \cite{Giddings15,Donnelly16} have argued that a quantum gravitational theory cannot possess any local observables, as any particle-creating operator will obtain a gravitational dressing, which extends infinitely far and is impossible to neutralize, as soon as gravity is turned on. In particular, they show that if diffeomorphisms are treated as gauge transformations, then there are no gauge-invariant local observables even in the first order in the gravitational coupling. These results are not easy to evade, because the lack of local diffeomorphism-invariant observables is true even in the classical theory \cite{Carlip12}. Indeed, it is intuitively clear that local observables cannot be invariant under arbitrary spacetime transformations. However, here it is important to distinguish between partial and complete observables, as defined by Rovelli \cite{Rovelli02}. There is no actual need to require that the local algebras we consider in this paper be invariant under spacetime transformations. We can allow spacetime transformations to have a non-trivial isomorphic action on the local algebras, mapping them to each other, as is the case also in generally covariant QFT \cite{Brunetti03}. In this case, we are considering algebras of \emph{partial} observables (instead of complete observables), which can be localized even if the theory has diffeomorphism invariance \cite{Rovelli02}. While the local partial observable algebras transform under spacetime transformations, so do the local states, thus keeping the expectation values (the actual predictable experimental data) invariant. On the other hand, the technical arguments in \cite{Torre93, Marolf15,Giddings15,Donnelly16} as we understand them essentially require that some geometric quantities (e.g., the metric) are fundamental dynamical variables in the theory. In contrast, in our view spacetime geometry need not be necessarily directly observable, but it may be possible to understand it as an effective description of the statistical properties of quantum states of matter and radiation (excluding gravity). From this perspective, the non-local aspects of gravitational interactions may arise from the non-local properties of quantum statistics, rather than the non-locality of observables.

In this paper, we mostly draw inspiration from the works \cite{Bannier94,Keyl96,Keyl98}, which develop methods to extract spacetime structure from the net of local operator algebras in the algebraic QFT setting. The main idea of our approach is to assume that the set of observable algebras associated to local regions is somehow already provided to us, which we then use to understand the necessary properties of a compatible spacetime topology. However, the ultimate goal in this direction of research would be to extract spacetime structure directly from the algebraic and statistical properties of observables. (See, e.g., \cite{Bannier94,Keyl96,Keyl98,CorichiRyanSudarsky02,SummersWhite03,BertozziniContiLewkeeratiyutkul10,Aguilar12,Cao16} for a small subset of works in this direction.) In order to guarantee background-independence without explicit diffeomorphism-invariance, it should be possible to formulate the theory in a way that does not directly refer to spacetime geometry, but only to the algebraic and statistical relations between quantum operators, while the effective spacetime geometry is extracted a posteriori. With this goal in mind, in \cite{Raasakka16} we formulated a spacetime-free framework for quantum theory. In particular, there should be some inherent way to define causal relations and locality without the reference to some background geometry. The extraction of locality from dynamics has been explored, for example, by Cotler et al. \cite{Cotler17}, but a totally satisfactory background-independent approach is still lacking. If such an approach was succesful, then the effective notion of locality would be the one to use in the definition of local observable algebras.\\

Let us then summarize the contents of the paper. In Section \ref{sec:aqft} we review the formalism of algebraic QFT, and work out the relationship between spacetime topology and the net of local operator algebras. In Section \ref{sec:findim} we modify the formalism introduced in the previous section by replacing the infinite-dimensional algebras of QFT by finite-dimensional ones, and explore the implications of this change on the associated spacetime topology. We find that spacetime topology must be significantly modified when the finite-dimensionality of local algebras is manifest, presumably at the Planck scale. In particular, there must exist minimal spacetime regions, although continuous spacetime transformations are still possible. Accordingly, Planck scale topology of spacetime turns out have features of both discreteness and continuity. We finish with a summary and some final remarks in Section \ref{sec:summary}.

\section{Algebraic quantum field theory and spacetime topology}\label{sec:aqft}
All observations of spacetime properties are performed in practice by studying the propagation of quantum fields in spacetime. Therefore, the operational information about spacetime geometry must be encoded into the structure of QFT. On the other hand, QFT models are usually built on top of a fixed background geometry. In this section, we will study the exact relationship between spacetime topology and QFT, and show how spacetime topology can be recovered from the algebraic properties of QFT. In particular, we will adopt the view that the physical meaning to a spacetime region is given exactly by the observables localized in that region.

The starting point for the algebraic formulation of QFT is that we associate to any causally convex\footnote{A spacetime region is causally convex if it contains entirely any causal curve between any two of its points.} open spacetime region $\mc{O}$ with compact closure an algebra of operators $\mf{A}_\mc{O}$, the \emph{local observable algebra}, whose self-adjoint elements are the observables localized inside the region $\mc{O}$. (See, e.g., \cite{Fewster19} for a recent accessible review of algebraic QFT, or \cite{Haag96,BarFredenhagen09} for more thorough textbook expositions.) As already mentioned, in QFT the local observable algebras are infinite-dimensional and in physically relevant models, more specifically, hyperfinite type $\text{III}_1$ von Neumann factors \cite{Buchholz87}. As there is only one hyperfinite type $\text{III}_1$ von Neumann factor up to isomorphisms, it is really the inclusion relations of local algebras, which encode the physical properties of a QFT model. The inclusion relations $\mf{A}_{\mc{O}_1}\subset\mf{A}_{\mc{O}_2}$ of the local algebras must obviously reflect spacetime topology, since any observation localized in $\mc{O}_1$ must also be localized inside $\mc{O}_2$ if $\mc{O}_1 \subset \mc{O}_2$. More specifically, if $\mc{O}_1 \subset \mc{O}_2$ is a proper inclusion, then in physically relevant models $\mf{A}_{\mc{O}_1}$ is a proper unital subalgebra of $\mf{A}_{\mc{O}_2}$. The partially ordered set of operator algebras index by spacetime regions is called \emph{the net of local algebras}.

The correspondence between local algebras and spacetime regions is not one-to-one as such, because any two regions with the same causal completion are associated with the same local algebra due to the causal dynamics of the field(s).
\begin{definition}
	The \emph{causal complement} $\mc{O}^c$ of a spacetime region $\mc{O}$ consists of all the points, which cannot be connected to $\mc{O}$ by a causal (i.e., everywhere light- or time-like) curve. The \emph{causal completion} of a spacetime region $\mc{O}$ is obtained as the double-complement $(\mc{O}^c)^c =: \mc{O}^{cc}$. A spacetime region $\mc{O}$ is called \emph{causally complete} if $\mc{O}^{cc} = \mc{O}$.
\end{definition}
\begin{definition}
	A \emph{Cauchy slice} $\Sigma$ of a spacetime region $\mc{O}$ is a spacelike codimension-1 hypersurface in $\mc{O}$, such that any inextendible causal curve in $\mc{O}$ intersects $\Sigma$ exactly once.
\end{definition}
Classically, initial data on any Cauchy slice $\Sigma$ of $\mc{O}$ determines the state of the field system in the whole of $\mc{O}^{cc}$ if the system obeys causal (hyperbolic) evolution equations and $\mc{O}^{cc}$ is globally hyperbolic. On the algebraic level this implies that, due to the dynamical evolution of the system, spacetime regions sharing the same Cauchy slice are associated to the same local algebra. Accordingly, we may restrict to consider causally complete spacetime regions in order to have one-to-one correspondence between spacetime regions and local algebras.

Both the set of causally complete open spacetime regions with compact closure and the set of local operator algebras can be seen to form order-theoretical lattices. The rest of the paper relies significantly on the theory of lattices. For lattice theory basics, we refer the reader to \cite{Birkhoff48,Davey02}.
\begin{definition}
	A \emph{lattice} $\mc{L}$ is a partially ordered set, in which any two elements $A,B\in\mc{L}$ have a least upper bound (\emph{join}) $A\vee B\in\mc{L}$ and a greatest lower bound (\emph{meet}) $A\wedge B\in\mc{L}$, defined in terms of the ordering as
	\ba
		A\vee B &= \inf\{C\in\mc{L} : C\geq A,\ C\geq B\} \,, \nonumber\\
		A\wedge B &= \sup\{C\in\mc{L} : C\leq A,\ C\leq B\} \,.
	\ea
	In a \emph{complete} lattice every subset $\mc{K}\subset\mc{L}$ has a greatest lower bound and a least upper bound.
\end{definition} 
In particular, we will consider the following two lattices:
\begin{definition}
	Let $\mc{O}\subset\mc{M}$ be a causally complete open subset with compact closure of a globally hyperbolic Lorentzian manifold $\mc{M}$. Then, $\mc{L}_\text{cc}(\mc{O})$ is the lattice, whose elements are the causally complete open subsets $\mc{O}_1\subset \mc{O}$. The order-relation in $\mc{L}_\text{cc}(\mc{O})$ is given by $\mc{O}_1 < \mc{O}_2\ \Leftrightarrow\ \mc{O}_1 \subsetneq \mc{O}_2$. The join and the meet are given, respectively, by $\mc{O}_1\vee\mc{O}_2 = (\mc{O}_1\cup\mc{O}_2)^{cc}\in\mc{L}_\text{cc}(\mc{O})$ and $\mc{O}_1\wedge\mc{O}_2 = \mc{O}_1\cap\mc{O}_2\in\mc{L}_\text{cc}(\mc{O})$ for any $\mc{O}_1,\mc{O}_2\in\mc{L}_\text{cc}(\mc{O})$. (The intersection of two causally complete regions is again causally complete, but the same is not true for the union.) $\mc{L}_\text{cc}(\mc{O})$ is a complete lattice with the least element $\emptyset$ and the greatest element $\mc{O}$.
\end{definition}
\begin{definition}
	Let $\mc{O}\subset\mc{M}$ be as above. Then, $\mc{L}_\text{alg}(\mc{O})$ is the lattice, whose elements are the local observable algebras $\mf{A}_{\mc{O}_1}$ associated to the elements of $\mc{L}_\text{cc}(\mc{O})$ through the map $\mf{A}: \mc{O}_1 \mapsto \mf{A}_{\mc{O}_1}$. The order-relation in $\mc{L}_\text{alg}(\mc{O})$ is given by $\mf{A}_{\mc{O}_1} < \mf{A}_{\mc{O}_2}\ \Leftrightarrow\ \mf{A}_{\mc{O}_1} \subset \mf{A}_{\mc{O}_2}$ as a proper unital subalgebra. The join and the meet in $\mc{L}_\text{alg}(\mc{O})$ are given, respectively, by $\mf{A}_{\mc{O}_1}\vee\mf{A}_{\mc{O}_2} = \mf{A}_{(\mc{O}_1\cup\mc{O}_2)^{cc}}\in\mc{L}_\text{alg}(\mc{O})$ and $\mf{A}_{\mc{O}_1}\wedge\mf{A}_{\mc{O}_2} = \mf{A}_{\mc{O}_1\cap\mc{O}_2} \in\mc{L}_\text{alg}(\mc{O})$ for any $\mf{A}_{\mc{O}_1},\mf{A}_{\mc{O}_2}\in\mc{L}_\text{alg}$. $\mc{L}_\text{alg}(\mc{O})$ is a complete lattice with the least element $\mf{A}_\emptyset \cong \mb{C}$ and the greatest element $\mf{A}_\mc{O}$.\footnote{The least element $\mf{A}_\emptyset \cong \mb{C}$ is the algebra generated by the common unit element shared by all the local algebras. The completeness of $\mc{L}_\text{alg}(\mc{O})$ can be shown, e.g., by noticing that $\phi: \mf{B} \mapsto \inf\{ \mf{A}_{\mc{O}_1} \in \mc{L}_\text{alg}(\mc{O}) : \mf{B} \subset \mf{A}_{\mc{O}_1}\}$, where $\mf{B}\subset\mf{A}_\mc{O}$ is any subfactor of type $\text{III}_1$ (not necessarily local), is a closure operation in the lattice of all subfactors of type $\text{III}_1$ ordered by inclusion.}
\end{definition}
\begin{assumption}\label{ass:iso}
	The two lattices  $\mc{L}_\text{cc}(\mc{O})$ and $\mc{L}_\text{alg}(\mc{O})$ are isomorphic, $\mc{L}_\text{cc}(\mc{O}) \cong \mc{L}_\text{alg}(\mc{O})$.
\end{assumption}
Since the map $\mf{A}: \mc{L}_\text{cc}(\mc{O}) \rightarrow \mc{L}_\text{alg}(\mc{O}), \mc{O}_1 \mapsto \mf{A}_{\mc{O}_1}$ is bijective, this assumption essentially requires that $\mc{O}_1 < \mc{O}_2\ \Leftrightarrow\ \mf{A}_{\mc{O}_1} < \mf{A}_{\mc{O}_2}$, so that $\mf{A}$ gives an order-isomorphism between the two lattices. This property that larger causally complete spacetime regions have larger algebras is generally satisfied by all physical QFT models, as far as we know. It is also the mathematical formulation of the idea that a spacetime region is operationally defined by the observables that are localized in it. Therefore, Assumption \ref{ass:iso} seems physically well-motivated.

Drawing inspiration from \cite{Bannier94,Keyl96,Keyl98}, we will now study some details of the recovery of spacetime topological structure from the lattice of local algebras.\\

{\bf Open regions and topology.} Since the causally complete regions form a base for the spacetime topology, an arbitrary open region can be expressed as the union of some set of elements in $\mc{L}_\text{cc}(\mc{O})$, by definition. The meet operation in $\mc{L}_\text{cc}(\mc{O})$, however, is not directly the union of regions, but its causal completion. Therefore, we cannot directly use the meet operation in $\mc{L}_\text{cc}(\mc{O})$ to define arbitrary spacetime regions. Instead, we may identify arbitrary spacetime regions as certain subsets in $\mc{L}_\text{cc}(\mc{O})$. To that end, we need a recall few more basic definitions from lattice theory, and prove a couple of propositions.
\begin{definition}
	Let $\mc{L}$ be a complete lattice with the least element $0$, and $\omega\subset \mc{L}$ a subset of elements. $\omega$ is called a \emph{down-set} in $\mc{L}$ if $A\in \omega$ and $B<A$ imply $B\in \omega$. The down-sets of $\mc{L}$ constitute themselves a complete lattice $\mc{L}_{ds}(\mc{L})$, ordered by inclusion, with the least element $\{0\}$ and the greatest element $\mc{L}$. The meet and the join in $\mc{L}_{ds}(\mc{L})$ are given by the set-theoretic union and intersection, respectively, i.e., $\omega_1\vee \omega_2 = \omega_1 \cup \omega_2$ and $\omega_1\wedge \omega_2 = \omega_1\cap \omega_2$.\footnote{Notice that the union and the intersection of down-sets is again a down-set.}
\end{definition}
\begin{definition}
	Let $\mc{L}$ be a complete lattice with the least element $0$, and $\mc{L}_{ds}(\mc{L})$ the lattice of its down-sets. The \emph{pseudo-complement} $\omega^*\in \mc{L}_{ds}(\mc{L})$ of an element $\omega\in\mc{L}_{ds}(\mc{L})$ is given by
	\ba
		\omega^* = \sup\{\omega'\in\mc{L}_{ds}(\mc{L}) : \omega\cap \omega' = \{0\}\} \,.
	\ea
	The double-pseudo-complementation $\omega\mapsto (\omega^*)^* =: \omega^{**}$ defines a closure operation in $\mc{L}_{ds}(\mc{L})$. The \emph{$**$-closed down-sets} in $\mc{L}_{ds}(\mc{L})$, which satisfy $\omega^{**} = \omega$, form a complete lattice $\mc{L}_{ds}^{**}(\mc{L})$ with the meet $\omega_1 \vee \omega_2 = (\omega_1 \cup \omega_2)^{**}$ and the join $\omega_1 \wedge \omega_2 = \omega_1 \cap \omega_2$.
\end{definition}
\begin{definition}
	A complete lattice $\mc{L}$ is called a \emph{frame} (also a \emph{complete Heyting algebra} or a \emph{locale}, depending on the context) if the distributive law $A \wedge (\vee_i B_i) = \vee_i (A \wedge B_i)$ holds for arbitrary collections of elements $\{B_i\}_i\subset \mc{L}$. An equivalent condition is that $\mc{L}$ satisfies the finite distributive law $A \wedge (B \vee C) = (A \wedge B) \vee (A \wedge C)$ for all $A,B,C\in\mc{L}$, and the map $A \mapsto A\wedge B$ preserves the suprema of directed sets in $\mc{L}$ for all $B\in\mc{L}$.
\end{definition}
\begin{proposition}\label{prop:frm}
	Let $\mc{L}$ be a complete lattice. Then the lattice $\mc{L}_{ds}^{**}(\mc{L})$ of $**$-closed down-sets in $\mc{L}$ is a frame.\footnote{Even though this proposition seems like a basic result in lattice theory, we were not able to find it in the literature. Accordingly, we include the proof here for completeness.}
\end{proposition}
\begin{proof}
	Let us first show that if $x\in \omega^{**}$ such that $x > 0$, then there exists $y\in\omega$ such that $0<y\leq x$. First we note that $x$ cannot belong to $\omega^*$ if there exists $y\in\omega$ such that $0<y\leq x$, because then also $y\in\omega^*$, which contradicts $\omega\cap\omega^*=\{0\}$. In fact,
	\ba
		\omega^* = \{x\in\mc{L} : \nexists y\in\omega \text{ s.t. } 0<y\leq x\}\,,
	\ea
	since this property defines a down-set, which cannot be further enlarged without violating the condition $\omega\cap\omega^*=\{0\}$. Since $\omega^{**}\cap\omega^*=\{0\}$, the claim follows.
	
	Let us then show that for all $\omega_1,\omega_2\in\mc{L}_{ds}(\mc{L})$
	\ba
		\omega_1 \cap \omega_2 = \{0\}\ \Leftrightarrow\ \omega_1^{**} \cap \omega_2 = \{0\} \,.
	\ea
	As we just showed, for any $x\in\omega_1^{**}$ there exists $y\in\omega_1$ such that $0<y\leq x$. Now, since $\omega_1^{**}\cap \omega_2$ is a down-set, if $x\in \omega_1^{**}\cap \omega_2$ then also $y\in \omega_1^{**}\cap \omega_2$, and consequently $y\in \omega_1\cap \omega_2$. Thus,
	\ba
		\omega_1^{**}\cap \omega_2 \neq \{0\}\ \Rightarrow\ \omega_1\cap \omega_2 \neq \{0\} \,.
	\ea
	By negating this implication, we get $\omega_1 \cap \omega_2 = \{0\}\ \Rightarrow\ \omega_1^{**} \cap \omega_2 = \{0\}$. The implication in the other direction is trivial, as $\omega_1\subset\omega_1^{**}$ and the intersection operation is monotonic.
	
	Let us show next that $\omega_1 \cap \omega_2^{**} = (\omega_1 \cap \omega_2)^{**}$ for all $\omega_1\in\mc{L}_{ds}^{**}(\mc{L})$ and $\omega_2\in\mc{L}_{ds}(\mc{L})$. We get an equivalent statement $(\omega_1 \cap \omega_2^{**})^* = (\omega_1 \cap \omega_2)^{*}$ by taking pseudo-complements on both sides, since the pseudo-complement is unique for elements in $\mc{L}_{ds}^{**}(\mc{L})$. Here,
	\ba
		(\omega_1 \cap \omega_2^{**})^* &= \sup\{\omega'\in\mc{L}_{ds}(\mc{L}) : (\omega_1\cap \omega_2^{**})\cap \omega' = \{0\}\} \,,\nn
		(\omega_1 \cap \omega_2)^{*} &= \sup\{\omega'\in\mc{L}_{ds}(\mc{L}) : (\omega_1\cap \omega_2)\cap \omega' = \{0\}\} \,.\label{eq:closets}
	\ea
	By the previous result
	\ba
		\omega_2^{**} \cap (\omega_1\cap \omega') = \{0\}\ \Leftrightarrow\ \omega_2 \cap (\omega_1\cap \omega') = \{0\} \,.
	\ea
	Thus, the two sets in (\ref{eq:closets}) are the same, and hence $(\omega_1 \cap \omega_2^{**})^* = (\omega_1 \cap \omega_2)^{*}$.
	
	It then follows immediately from $\omega_1 \cap \omega_2^{**} = (\omega_1 \cap \omega_2)^{**}$ that $\mc{L}_{ds}^{**}(\mc{L})$ satisfies the finite distributivity law $\omega_1 \wedge (\omega_2 \vee \omega_3) = (\omega_1 \wedge \omega_2) \vee (\omega_1 \wedge \omega_3)$ for all $\omega_1,\omega_2,\omega_3\in\mc{L}_{ds}^{**}(\mc{L})$:
\ba
	\omega_1 \wedge (\omega_2 \vee \omega_3) &= \omega_1 \cap (\omega_2 \cup \omega_3)^{**} \nn
	&= (\omega_1 \cap (\omega_2 \cup \omega_3))^{**} \nn
	&= ((\omega_1 \cap \omega_2) \cup (\omega_1 \cap \omega_3))^{**} \nn
	&= (\omega_1 \wedge \omega_2) \vee (\omega_1 \wedge \omega_3) \,.
\ea
Since the meet operation in $\mc{L}_{ds}^{**}(\mc{L})$ is just the set-theoretical intersection, and the order relation in $\mc{L}_{ds}^{**}(\mc{L})$ is given by the set-theoretical inclusion, the preservation of suprema is immediate.
\end{proof}

Now, let us define the lattice of open subsets in $\mc{O}$, which captures the topology of $\mc{O}$. Open subsets of a topological space equipped with the set-theoretical union and intersection as the join and the meet operations, respectively, form also a frame. Thus, a frame is often considered to define a topological space in terms of open sets without referring to points \cite{Johnstone82,Picado12}.
\begin{definition}
	The set $\mc{T}(\mc{O})$ of open subsets in $\mc{O}$ constitute a frame when equipped with the ordering relation $O_1 < O_2\ \Leftrightarrow\ O_1 \subsetneq O_2$ for $O_1,O_2\in\mc{T}(\mc{O})$. The join and the meet in $\mc{T}(\mc{O})$ are given by the usual set-theoretical union and intersection, i.e., $O_1\vee O_2 = O_1\cup O_2$ and $O_1\wedge O_2 = O_1\cap O_2$, respectively.
\end{definition}
\begin{proposition}\label{prop:top}
	Let $\mc{O}$ be a causally complete open spacetime region with a compact closure. Let $\mc{T}(\mc{O})$ be as defined above (i.e., the topology of $\mc{O}$), and $\mc{L}_{ds}^{**}(\mc{L}_\text{cc}(\mc{O}))$ the lattice of $**$-closed down-sets of $\mc{L}_\text{cc}(\mc{O})$. Then, $\mc{T}(\mc{O}) \cong \mc{L}_{ds}^{**}(\mc{L}_\text{cc}(\mc{O}))$.
\end{proposition}
\begin{proof}
	Let $\phi: \mc{T}(\mc{O}) \rightarrow \mc{L}_{ds}(\mc{L}_\text{cc}(\mc{O}))$ map the open spacetime region $O\in \mc{T}(\mc{O})$ to the down-set $\phi(O)\in\mc{L}_{ds}(\mc{L}_\text{cc}(\mc{O}))$, which is the set of causally complete regions contained in the region $O$. $\phi$ is clearly injective (assuming the usual Hausdorff topology on $\mc{O}$). We want to show that $\phi$ is also surjective to the set $\mc{L}_{ds}^{**}(\mc{L}_\text{cc}(\mc{O}))$.
	
	The left-inverse of $\phi$ is given by $\chi: \mc{L}_{ds}(\mc{L}_\text{cc}(\mc{O})) \rightarrow \mc{T}(\mc{O}), \omega \mapsto \cup_{\mc{O}_1\in \omega}\mc{O}_1$, i.e., $\chi\circ\phi = \id_{\mc{T}(\mc{O})}$. We will show in the following that $\phi\circ\chi = (\cdot)^{**}$, which proves that $\phi$ is surjective to $\mc{L}_{ds}^{**}(\mc{L}_\text{cc}(\mc{O}))$. Now, first of all, notice that $\omega_1\cap \omega_2 = \{\emptyset\}\ \Leftrightarrow\ \chi(\omega_1)\cap\chi(\omega_2) = \emptyset$ for all $\omega_1,\omega_2\in\mc{L}_{ds}(\mc{L}_\text{cc}(\mc{O}))$: If $\chi(\omega_1)\cap\chi(\omega_2) \neq \emptyset$, there would be some non-empty causally complete region contained in $\chi(\omega_1)\cap\chi(\omega_2)$, which would belong to $\omega_1\cap \omega_2$. Also, if $\omega_1\cap \omega_2 \neq \{\emptyset\}$, it is clear that $\chi(\omega_1)\cap\chi(\omega_2)\neq\emptyset$. Accordingly, $\phi(\text{cl}(\chi(\omega))^\perp) = \omega^*$, where $\text{cl}(O)^\perp\in\mc{T}(\mc{O})$ is the complement of the closure of $O\in\mc{T}(\mc{O})$: Clearly, $\text{cl}(\chi(\omega))^\perp$ is the largest open set in $\mc{T}(\mc{O})$, which does not overlap with $\chi(\omega)$. Consequently, we get $\omega^{**} = \phi(\text{cl}(\chi(\phi(\text{cl}(\chi(\omega))^\perp)))^\perp) = \phi(\chi(\omega))$.
	
	Finally, it is easy to see that $\omega_1 < \omega_2\ \Leftrightarrow\ \chi(\omega_1) < \chi(\omega_2)$ for any $\omega_1,\omega_2 \in\mc{L}_{ds}^{**}(\mc{L}_\text{cc}(\mc{O}))$, which shows that $\phi$ is a lattice isomorphism.
\end{proof}
Since $\mc{L}_\text{cc}(\mc{O}) \cong \mc{L}_\text{alg}(\mc{O})$, and the lattice $\mc{T}(\mc{O})$ of open subsets of $\mc{O}$ captures the topology of $\mc{O}$, we may recover the topology of $\mc{O}$ as the frame $\mc{L}_{ds}^{**}(\mc{L}_\text{alg}(\mc{O}))$ of $**$-closed down-sets in the lattice $\mc{L}_\text{alg}(\mc{O})$ of local observable algebras of QFT.\\

{\bf Points.} Spacetime points $x\in\mc{O}$ are in one-to-one correspondence with the sets of open spacetime regions $\mc{K}_x = \{O_1 \in \mc{T}(\mc{O}) : x \in O_1\}$. The sets $\mc{K}_x$ can be uniquely characterized order-theoretically as completely prime filters in the frame $\mc{T}(\mc{O})\cong\mc{L}_{ds}^{**}(\mc{L}_\text{cc}(\mc{O}))$ \cite{Johnstone82,Picado12}.
\begin{definition}
	Let $\mc{L}$ be a lattice. A non-empty subset $F\subset\mc{L}$ is called a \emph{filter} if
	\begin{itemize}
		\item[(1)] $F$ is an \emph{up-set}, i.e., $x\in F,\ y\geq x\ \Rightarrow\ y\in F$, and
		\item[(2)] for all $x,y\in F$ there exists $z\in F$ s.t.\ $z\leq x$ and $z\leq y$.
	\end{itemize}
	A filter $F$ is \emph{completely prime} if for any subset $B\subset\mc{L}$ the following implication holds:
	\ba
		\vee B \in F\ \Rightarrow\ \exists x\in B \text{ s.t. } x\in F\,.
	\ea
 \end{definition}
Since $\mc{L}_\text{cc}(\mc{O}) \cong \mc{L}_\text{alg}(\mc{O})$, we may define `spacetime points' just as well as the completely prime filters in $\mc{L}_{ds}^{**}(\mc{L}_\text{alg}(\mc{O}))\cong\mc{T}(\mc{O})$.\\

Remarkably, the above results provide the inverse to the initial construction of the net of local algebras starting from a spacetime region $\mc{O}$ (which could also be the whole spacetime). We can indeed recover the topology of spacetime from the net of local algebras, and thus give it an operationally well-defined meaning. However, the structure of the lattice of local algebras $\mc{L}_\text{cc}(\mc{O})$ is significantly modified in the locally finite-dimensional case, as we will see in the next section, and therefore local finite-dimensionality implies definite modifications to spacetime topology at scales where it is manifest.

\section{Implications of local finite-dimensionality for topology}\label{sec:findim}
Let us now suppose that the local algebras $\mf{A}_\mc{O}$ are, in fact, finite-dimensional. To be more precise, we will assume the following.
\begin{assumption}\label{ass:fin}
	Let $\mc{O}$ be a spacetime region inside which any Cauchy slice has a finite spatial volume. Then the corresponding local observable algebra $\mf{A}_\mc{O}$ is a finite-dimensional factor.
\end{assumption}	
Any local observable algebra is thus isomorphic to the algebra of $n$-by-$n$ complex matrices, from hereon denoted by $\mb{M}_n$, for some $n \in\mb{N}$. We will again consider the ordering relation
\ba
	\mf{A}_1 < \mf{A}_2\ \Leftrightarrow\ \mf{A}_1 \subset \mf{A}_2\ \text{as a proper unital subalgebra}
\ea
for the local algebras. As already argued above, the physical meaning to a spacetime region is given by the observations that can be made inside that region. Therefore, we expect Assumption \ref{ass:iso} to remain valid, and postulate that $\mc{L}_\text{cc}(\mc{O}) \cong \mc{L}_\text{alg}(\mc{O})$ still holds for the lattice $\mc{L}_\text{cc}(\mc{O})$ of causally complete spacetime subregions, which provides a base for the topology of $\mc{O}$. Thus, the topology of $\mc{O}$ may still be obtained as the frame $\mc{L}_{ds}^{**}(\mc{L}_\text{alg}(\mc{O}))$.

\subsection{Lattice of all subfactors}
What kind of modifications does the change to finite-dimensional factors imply for the lattice of local algebras? To set the stage, let us first consider the lattice $\mc{L}_{sub}(\mf{A})$ of all subfactors of a finite-dimensional factor $\mf{A} \cong \mb{M}_n$. Some basic properties of the lattice of \emph{local} subfactors will follow directly from the properties of $\mc{L}_{sub}(\mf{A})$ together with some simple physically motivated assumptions about the local subfactors. Notice that even though the structure of $\mc{L}_{sub}(\mf{A})$ is fairly easy to understand, it is still not totally trivial. For example, it is not a finite lattice, since there are continuous families of subfactors given by unitary transformations.

Let us recall further basic definitions from lattice theory.
\begin{definition}
	Let $\mc{L}$ be a complete lattice with the greatest element $1$ and the least element $0$. $\mc{L}$ has \emph{finite length} if any chain of elements $0 < A_1 < A_2 < \ldots < 1$ consists of only a finite number of elements.
\end{definition}
\begin{definition}
	Let $\mc{L}$ be a complete lattice with the greatest element $1$ and the least element $0$. $\mc{L}$ is \emph{complemented} if every element $A\in\mc{L}$ has a complementary element $A'\in\mc{L}$ such that $A\vee A' = 1$ and $A\wedge A' = 0$.
\end{definition}
\begin{definition}
	An \emph{atomic} element $A\in\mc{L}$ of a lattice $\mc{L}$ with the least element $0$ is such that there does not exist another element $B \in \mc{L}$ such that $0 < B < A$. A lattice $\mc{L}$ is called \emph{atomic} if for every non-atomic element $B \in \mc{L}$ there exists an atomic element $A\in\mc{L}$ such that $A<B$. A lattice is called \emph{atomistic} if every element $B \in \mc{L}$ is the join of a set of atomic elements.
\end{definition}

Next we explore some basic properties of $\mc{L}_{sub}(\mf{A})$.
\begin{proposition}
	The subfactors of a finite-dimensional factor $\mf{A}$ form a complete atomistic lattice $\mc{L}_{sub}(\mf{A})$ of finite length, when equipped with the ordering relation
	\ba
		\mf{A}_1 < \mf{A}_2\ \Leftrightarrow\ \mf{A}_1 \subset \mf{A}_2\ \textrm{\emph{as a proper unital subalgebra}}
	\ea
	for any two subfactors $\mf{A}_1,\mf{A}_2 \subset \mf{A}$.
\end{proposition}
\begin{proof}
	Let us again note that any finite-dimensional factor $\mf{A}$ is isomorphic to $\mb{M}_n$ for some $n \in\mb{N}$.  Moreover, any unital inclusion of a full matrix algebra to another is of the form
	\ba\label{eq:incl}
		\mb{M}_m \mapsto U(\mb{M}_m \otimes \1_n)U^* \subset \mb{M}_{mn} \,,
	\ea
	where $m,n\in\mb{N}$, $U$ is some unitary in $\mb{M}_{mn}$, and the tensor product is defined with respect to some arbitrary basis. Accordingly, if $\mf{A}_1 < \mf{A}_2$ for $\mf{A}_1\cong\mb{M}_m$ and $\mf{A}_2\cong\mb{M}_{m'}$, then $m'$ must be divisible by $m$. The least element is the trivial subfactor $\mb{C}\1_n = \{c\1_n : c\in\mb{C}\} \cong \mb{C}$, where $\1_n$ denotes the $n$-by-$n$ identity matrix, and the greatest element is $\mf{A}$ itself. Thus, the maximum length of a chain of elements in $\mc{L}_{sub}(\mf{A})$ for $\mf{A}\cong\mb{M}_n$ is the number of prime factors in $n$, which is obviously finite. This shows that $\mc{L}_{sub}(\mf{A})$ has finite length. Completeness of $\mc{L}_{sub}(\mf{A})$ is trivial.
	
	Complementation on $\mc{L}_{sub}(\mf{A})$ is given by taking the commutant,
	\ba
		\mf{A}_1' = \{a'\in\mf{A} : [a,a']= 0\ \forall\ a\in\mf{A}_1\} \,.
	\ea
	for all $\mf{A}_1\in \mc{L}_{sub}(\mf{A})$. The lattice complement properties are trivial to check, since for finite-dimensional algebras $\mf{B}\vee \mf{C} \cong \mf{B}\otimes\mf{C}$ if $\mf{B},\mf{C}\subset\mf{A}$ are mutually commuting subfactors.
	
	As already noted, a unital inclusion $\mb{M}_m \subset \mb{M}_{m'}$ of the form (\ref{eq:incl}) is only possible if $m'$ is divisible by $m$. Accordingly, $\mb{M}_p$ does not contain any proper non-trivial subfactors for $p$ prime. $\mb{M}_n$ for any $n\in\mb{N}$ can be factorized into a tensor product as $\mb{M}_n \cong \otimes_{k} \mb{M}_{p_k}$, where $\{p_k\}_k$ are the prime factors of $n$ with multiplicity. Accordingly, $\mb{M}_n$ has subfactors $U(\mb{M}_{p_k} \otimes \1_{n/p_k})U^*$, where $U$ is any unitary in $\mb{M}_n$. (Notice that a reordering of tensor product factors is also a unitary operation.) These are atomic elements in $\mc{L}_{sub}(\mf{A})$ for $\mf{A}\cong\mb{M}_n$, since they do not contain any other proper subfactors besides the trivial one $\mb{C}\1_n\cong\mb{C}$. Moreover, they provide all the atomic elements, since the tensor factorization is unique up to unitary transformations.
	
	Finally, let us show that $\mc{L}_{sub}(\mf{A})$ is atomistic. Any subfactor of $\mb{M}_n$ can be expressed as $U(\mb{M}_m \otimes \1_{n/m})U^*$ for some unitary $U$ in $\mb{M}_n$. Accordingly, it can be expressed in terms of the atomic elements as
	\ba
		U(\mb{M}_m \otimes \1_{n/m})U^* = U((\otimes_k\mb{M}_{q_k}) \otimes \1_{n/m})U^* \,,
	\ea
	where now $\{q_k\}$ are the prime factors of $m$. This is the smallest subfactor containing all the atomic subfactors
	\ba
		U(((\otimes_{k<l}\1_{q_k}) \otimes \mb{M}_{q_l} \otimes (\otimes_{k>l}\1_{q_k})) \otimes \1_{n/m})U^* = UV(\mb{M}_{q_l} \otimes \1_{n/q_l})V^*U^* \,,
	\ea
	where $V$ is a unitary, which moves the tensor product factors appropriately. Since any subfactor can be expressed as the least upper bound of atomic factors, the lattice $\mc{L}_{sub}(\mf{A})$ is atomistic. Notice, however, that the choice of atomic elements is not unique.
\end{proof}

\subsection{Lattice of local subfactors}
Now, let us consider the lattice $\mc{L}_\text{alg}(\mc{O})$ of local subfactors of a local observable algebra $\mf{A}_\mc{O}$. The elements of $\mc{L}_\text{alg}(\mc{O})$ constitute obviously a subset of the elements of $\mc{L}_{sub}(\mf{A}_\mc{O})$, the lattice of all subfactors in $\mf{A}_\mc{O}$. The challenge in trying to understand the structure of the lattice of \emph{local} subfactors is that, as we do not introduce a classical background geometry from the outset, without the reference to a classical background it is not clear which subfactors of $\mf{A}_\mc{O}$ are local. The question of what determines if a factor is local or not is rather intricate, as it should depend on the dynamics of the system. In \cite{Cotler17}, for example, a tensor factorization of the Hilbert space was deemed local if the Hamiltonian operator could be written as a sum of terms each coupling only a finite number of tensor product factors. However, for our purposes it suffices to make a couple of basic assumptions about the properties of $\mc{L}_\text{alg}(\mc{O})$:
\begin{assumption}\label{ass:int}
The intersection of two local subfactors is again a local subfactor, so that the subset of local subfactors is closed under the meet operation.
\end{assumption}
Note that the intersection of two subfactors is always a subfactor --- the locality part of this assumption is non-trivial. We expect this assumption to hold for any reasonable notion of local subsystems. In particular, it clearly holds for the notion of locality used in \cite{Cotler17}, when local subsystems are defined as arbitrary collections of tensor product factors in a local factorization of the Hilbert space.
\begin{assumption}\label{ass:com}
The commutant of a local subfactor is again a local subfactor.
\end{assumption}
This assumption relates to the Haag duality property in algebraic QFT, which states that $\mf{A}(\mc{O})'=\mf{A}(\mc{O}^c)$, and is generally satisfied in the vacuum sector \cite{Fewster19}. It is also satisfied by the notion of locality used in \cite{Cotler17}. Notice that if $\mf{B},\mf{C}\subset\mf{A}_\mc{O}$ are two local subfactors such that $\mf{B} \subset\mf{C}$, then the two assumptions together imply that the relative commutant
\ba
	\mf{B}'_\mf{C} = \{c\in\mf{C} : [b,c] = 0\ \forall\ b\in\mf{B}\} \subset \mf{C} \subset \mf{A}_\mc{O}
\ea
is also a local subfactor, because this is the intersection of $\mf{B}'$ and $\mf{C}$.

In the following, we will explore some basic properties of $\mc{L}_\text{alg}(\mc{O})$ in the finite-dimensional case.
\begin{proposition}\label{prop:comp}
$\mc{L}_\text{alg}(\mc{O})$ is a complete complemented lattice of finite length.
\end{proposition}
\begin{proof}
Assumption \ref{ass:int} implies that the elements of $\mc{L}_\text{alg}(\mc{O})$ form a topped $\cap$-structure in $\mc{L}_{sub}(\mf{A}_\mc{O})$, as considered in \cite{Davey02}. This is equivalent with the property that the surjective map
	\ba
		\text{cc}: \mc{L}_{sub}(\mf{A}_\mc{O}) \rightarrow \mc{L}_\text{alg}(\mc{O}) \subset \mc{L}_{sub}(\mf{A}_\mc{O}) ,\ \mf{B} \mapsto \inf\{\mf{C}\in \mc{L}_\text{alg}(\mc{O}) : \mf{A} \subset \mf{C}\} \,,
	\ea
	which maps an arbitrary subfactor $\mf{B}$ of $\mf{A}_\mc{O}$ to the smallest local subfactor of $\mf{A}_\mc{O}$ containing $\mf{B}$, is a closure operator in $\mc{L}_{sub}(\mf{A}_\mc{O})$. The fact that $\mc{L}_\text{alg}(\mc{O})$ is obtained from $\mc{L}_{sub}(\mf{A}_\mc{O})$ via a closure operator implies that $\mc{L}_\text{alg}(\mc{O})$ is complete (with the least element $\mb{C}$ and the greatest element $\mf{A}_\mc{O}$), when equipped with the same ordering relation.\footnote{Notice, however, that the join operation in $\mc{L}_\text{alg}(\mc{O})$ is not the same as in $\mc{L}_{sub}(\mf{A}_\mc{O})$, but is obtained through the closure from the latter.}
	
Assumption \ref{ass:com} makes $\mc{L}_\text{alg}(\mc{O})$ a complemented lattice, where the complement is given by the commutant, as in $\mc{L}_{sub}(\mf{A}_\mc{O})$.

The finite length of $\mc{L}_\text{alg}(\mc{O})$ follows immediately from the finite length of $\mc{L}_{sub}(\mf{A}_\mc{O})$.
\end{proof}

\begin{proposition}\label{prop:atom}
$\mc{L}_\text{alg}(\mc{O})$ is atomistic.
\end{proposition}
\begin{proof}
Let $\mf{B}_1 \in \mc{L}_\text{alg}(\mc{O})$ such that $\mf{B}_1 \neq \mb{C}$. Either $\mf{B}_1$ is atomic, or it can be decomposed as $\mf{B}_1 = \mf{B}_2 \vee (\mf{B}_{2})_{\mf{B}_1}'$ for some local subfactor $\mf{B}_2 < \mf{B}_1$. By Assumptions \ref{ass:int} and \ref{ass:com} also the relative commutant $(\mf{B}_{2})_{\mf{B}_1}'$ of $\mf{B}_2$ inside $\mf{B}_1$ is local, and so also $(\mf{B}_{2})_{\mf{B}_1}' < \mf{B}_1$. We can then iterate this splitting of local subfactors until we reach atomic elements. Since $\mc{L}_\text{alg}(\mc{O})$ has finite length, this will take only a finite number of steps.

Since $\mf{B}_1 \cong \mf{B}_2 \otimes (\mf{B}_{2})_{\mf{B}_1}'$ for any $\mf{B}_2 \subset \mf{B}_1$, the decomposition of $\mf{B}_1$ by such iterated splitting will result in a tensor product factorization of $\mf{B}_1$ into atomic local subalgebras. Accordingly, any non-trivial local subfactor $\mf{B}_1 \in \mc{L}_\text{alg}(\mc{O})$ is a join of atomic elements. Moreover, it is always possible to choose mutually commuting atomic subfactors, as the factors in a tensor product factorization commute.
\end{proof}

\subsection{Granularity of spacetime topology}
Now, according to our Assumption \ref{ass:iso} $\mc{L}_\text{alg}(\mc{O}) \cong \mc{L}_\text{cc}(\mc{O})$, and thus the lattice of causally complete spacetime regions $\mc{L}_\text{cc}(\mc{O})$ is also atomistic. This implies the existence of minimal spacetime regions, and thus the non-existence of spacetime points. We can still define arbitrary spacetime regions in a point-free manner as $**$-closed down-sets in $\mc{L}_\text{cc}(\mc{O})$, though, and consider the topology of spacetime based on this notion of regions. This leads naturally to a kind of point-free granularity of spacetime.

Let us articulate more clearly two basic consequences of the atomistic nature of $\mc{L}_\text{cc}(\mc{O})$ in terms of spacetime topology.
\begin{corollary}
	There exist atomic elements in $\mc{T}(\mc{O})\cong\mc{L}_{ds}^{**}(\mc{L}_\text{cc}(\mc{O}))$ corresponding to minimal spacetime regions: Any atomic element $\mc{A}\in\mc{L}_\text{cc}(\mc{O})$ gives rise to a minimal non-trivial $**$-closed down-set $A=\{\emptyset,\mc{A}\} \in \mc{L}_{ds}^{**}(\mc{L}_\text{cc}(\mc{O}))$. Let $\mc{A}_1,\mc{A}_2\in\mc{L}_\text{cc}(\mc{O})$ be two different atomic elements. Then $A_1\wedge A_2 = \{\emptyset\}$, i.e., the intersection of any two minimal spacetime regions is empty.
\end{corollary}
The triviality of the intersection of two minimal regions can be intuitively understood, as any such intersection would have to lead to a smaller region than the minimal regions, but there does not exist any such regions. The existence of minimal regions does not necessarily imply a discrete structure of spacetime in the usual sense, however, as there may exist continuous transformations which preserve the locality of regions. (Whether such transformations exist depends again on the notion of locality inherited from the dynamics of the system.) Unintuitively enough, if the system has continuous transformations which preserve localization, there exist unitary transformations arbitrarily close to the identity which produce a non-overlapping spacetime region, when applied to a minimal region. The non-overlapping property of minimal regions makes them somewhat point-like, although they have finite volume.
\begin{corollary}
	Let $\mc{O}_1\in\mc{L}_\text{cc}(\mc{O})$. Then $\mc{O}_1 = \vee_{i} \mc{A}_i$ for some finite set $\{\mc{A}_i\}_{i}$ of atomic elements, i.e., any causally complete region can be obtained as the causal completion of a union of a finite number of minimal spacetime regions.
\end{corollary}
The expression of a causally complete spacetime region $\mc{O}$ as a causal completion of a set of minimal regions is not unique, in general, but (especially in the presence of symmetries) there can be several different choices for the set $\{\mc{A}_i\}_{i}$ of atomic elements, e.g.,  corresponding to different local tensor product structures on $\mf{A}_\mc{O}$.\footnote{Cotler et al.\ \cite{Cotler17} have shown that generically there exists (at most) only one tensor product structure, which is $k$-local with respect to a given Hamiltonian. However, a generic Hamiltonian does not have any symmetries.}

We address the structure of the topology $\mc{T}(\mc{O})\cong\mc{L}_{ds}^{**}(\mc{L}_\text{alg}(\mc{O}))$ more carefully through the following propositions.
\begin{proposition}\label{prop:atom2}
	Let $\mc{L}$ be a complete atomistic lattice. Then $\mc{L}_{ds}^{**}(\mc{L})$ is a complemented atomistic frame, i.e., an atomistic Boolean algebra.
\end{proposition}
\begin{proof}
	$\mc{L}_{ds}^{**}(\mc{L})$ is a frame by Proposition \ref{prop:frm}. Let $\omega \in \mc{L}_{ds}^{**}(\mc{L})$, $B\in\omega$, and
	\ba
		\alpha_B = \{A \in\mc{L} : A \leq B,\ A \text{ atomic}\} \,.
	\ea
	Obviously, $B\in\omega\ \Rightarrow\ \alpha_B\subset \omega$, since $\omega$ is a down-set. However, the implication holds also in the other direction, $\alpha_B\subset \omega\ \Rightarrow\ B\in\omega$: Since $\alpha_B\cap\omega^* =\{0\}$, $B\in(\omega^*)^* = \omega$ by the definition of the pseudo-complement. As $\mc{L}$ is atomistic, we may then obtain the elements in $\omega$ by arbitrary joins of atomic elements in $\omega$, and therefore any $\omega \in \mc{L}_{ds}^{**}(\mc{L})$ is uniquely specified by its atomic elements. Let $\alpha_\omega\subset\mc{L}$ be a collection of atomic elements in $\omega \in \mc{L}_{ds}^{**}(\mc{L})$. Notice that $\alpha\cup\{0\}$ is a down-set. We then have $\omega = (\alpha\cup\{0\})^{**} = \vee_{A\in\alpha}\{0,A\}$, which is the smallest $**$-closed down-set containing all the atoms in $\alpha$. Thus, $\mc{L}_{ds}^{**}(\mc{L})$ is atomistic.
	
	To show that $\mc{L}_{ds}^{**}(\mc{L})$ is complemented, note that $\omega^*$ (as the largest down-set with trivial intersection $\{0\}$ with $\omega$) contains all the atoms not in $\omega$. Since $\omega$ and $\omega^*$ together contain all the atoms in $\mc{L}$, $(\omega \vee \omega^*)^* = \{0\}$ and so $\omega \vee \omega^* = (\omega \vee \omega^*)^{**} = \mc{L}$ for all $\omega \in \mc{L}_{ds}^{**}(\mc{L})$. Thus, $\omega^*$ is a proper complement, and $\mc{L}_{ds}^{**}(\mc{L})$ is complemented.
\end{proof}
\begin{proposition}\label{prop:atom3}
	Let $\mc{L}$ be a complete atomistic lattice, and $\mc{L}_\text{\emph{atom}}(\mc{L})$ the lattice of subsets of atoms in $\mc{L}$, ordered by inclusion. Then $\mc{L}_{ds}^{**}(\mc{L}) \cong \mc{L}_\text{\emph{atom}}(\mc{L})$.
\end{proposition}
\begin{proof}
	Let $\phi: \mc{L}_\text{atom}(\mc{L}) \rightarrow \mc{L}_{ds}^{**}(\mc{L})$ map $\alpha\in\mc{L}_\text{atom}(\mc{L})$ to the $**$-closed down-set $(\alpha\cup\{0\})^{**} = \vee_{A\in\alpha}\{0,A\}$. The inverse map $\phi^{-1}$ maps $\omega\in\mc{L}_{ds}^{**}(\mc{L})$ to the set of atoms contained in $\omega$. It is easy to see by the discussion in the proof of Proposition \ref{prop:atom2} that the join and meet are preserved,
	\ba
		\phi(\alpha_1 \vee \alpha_2) &= \phi(\alpha_1 \cup \alpha_2) = (\phi(\alpha_1) \cup \phi(\alpha_2))^{**} = \phi(\alpha_1) \vee \phi(\alpha_2) \,,\nn
		\phi(\alpha_1 \wedge \alpha_2) &= \phi(\alpha_1 \cap \alpha_2) = \phi(\alpha_1) \cap \phi(\alpha_2) = \phi(\alpha_1) \wedge \phi(\alpha_2) \,,
	\ea
	so $\phi$ is a lattice isomorphism.
\end{proof}
Since in the locally finite-dimensional case $\mc{L}_\text{alg}(\mc{O})$ is a complete atomistic lattice, by Proposition \ref{prop:atom2} the topology of spacetime $\mc{T}(\mc{O})\cong\mc{L}_{ds}^{**}(\mc{L}_\text{alg}(\mc{O}))$ is given by an atomistic Boolean algebra. Since by Proposition \ref{prop:atom3} $\mc{T}(\mc{O})\cong\mc{L}_\text{atom}(\mc{L}_\text{alg}(\mc{O}))$, the minimal local subalgebras in $\mc{L}_\text{alg}(\mc{O})$ correspond to indivisible non-overlapping chunks of spacetime, out of which any spacetime region can be constructed. The fact that $\mc{T}(\mc{O})$ is complemented implies that the spacetime regions $O\in\mc{T}(\mc{O})$ should be thought of as both open and closed (i.e., clopen). In particular, this prevents the definition of lower dimensional boundaries between regions as the intersection of their closures.
\begin{proposition}
	Completely prime filters in $\mc{L}_\text{\emph{atom}}(\mc{L})$ are of the form
	\ba
		F_A = \{B\subset\mc{L}_\text{\emph{atom}}(\mc{L}) : A\in B\}\,,\quad  A\in\mc{L} \text{ atomic.}
	\ea
\end{proposition}
\begin{proof}
	The claim follows directly from the definition of $\mc{L}_\text{atom}(\mc{L})$.
\end{proof}
Thus, the completely prime filters in $\mc{T}(\mc{O})\cong\mc{L}_{ds}^{**}(\mc{L}_\text{alg}(\mc{O}))\cong\mc{L}_\text{atom}(\mc{L}_\text{alg}(\mc{O}))$ are uniquely associated with the atomic elements in $\mc{L}_\text{alg}(\mc{O})$. As the completely prime filters in $\mc{L}_{ds}^{**}(\mc{L}_\text{alg}(\mc{O}))$ corresponded to spacetime points in the QFT case, this is another sense in which the minimal spacetime regions are point-like. However, at the same time, the spacetime regions are considered to occupy a finite volume. Therefore, one should rather have in mind a kind of cellular decomposition of spacetime into minimal regions. Accordingly, we infer a kind of point-free granularity of spacetime in the locally finite-dimensional case.

\section{Summary and discussion}\label{sec:summary}
In this paper we examined the consequences of local finite-dimensionality of physics for spacetime topology. For operational reasons, we postulated that the lattice $\mc{L}_\text{cc}(\mc{O})$ of causally complete spacetime subregions of a spacetime region $\mc{O}$ is isomorphic to the lattice $\mc{L}_\text{alg}(\mc{O})$ of local subfactors of the local observable algebra $\mf{A}_\mc{O}$ (Assumption \ref{ass:iso}). As the set of causally complete spacetime regions provides a base for spacetime topology, we were able to equate the lattice $\mc{T}(\mc{O})$ of arbitrary open subsets of $\mc{O}$ (i.e., the topology of $\mc{O}$) with the lattice $\mc{L}_{ds}^{**}(\mc{L}_\text{cc}(\mc{O}))$ of $**$-closed down-sets in $\mc{L}_\text{cc}(\mc{O})$ (Proposition \ref{prop:top}). Our method for deriving spacetime topology $\mc{T}(\mc{O})$ from the lattice local observable algebras $\mc{L}_\text{alg}(\mc{O})$ can then be expressed in one line as
\ba
	\mc{L}_\text{alg}(\mc{O}) \cong \mc{L}_\text{cc}(\mc{O}) \rightarrow \mc{L}_{ds}^{**}(\mc{L}_\text{cc}(\mc{O})) \cong \mc{T}(\mc{O}) \,.
\ea
We showed that in the case of QFT we recover the usual topology of spacetime from the lattice of local subfactors $\mc{L}_\text{alg}(\mc{O})$ as $\mc{T}(\mc{O})\cong \mc{L}_{ds}^{**}(\mc{L}_\text{alg}(\mc{O}))$.

We then went on to study the locally finite-dimensional case, and introduced three basic assumptions about the local observable algebras in a locally finite-dimensional model:
\begin{enumerate}
	\item Local observable algebras associated to spatially bounded spacetime regions are finite-dimen-sional factors (Assumption \ref{ass:fin}).
	\item The intersection of two local observable algebras is another local observable algebra (Assumption \ref{ass:int}).
	\item The commutant of a local observable algebra is another local observable algebra (Assumption \ref{ass:com}).
\end{enumerate}
Using these assumptions, we inferred that in the locally finite-dimensional case the lattice of causally complete spacetime regions is a complete complemented atomistic lattice of finite length (Propositions \ref{prop:comp} and \ref{prop:atom}). This implies the existence of minimal spacetime regions, and  thus the non-existence of spacetime points. More specifically, $\mc{T}(\mc{O})\cong\mc{L}_{ds}^{**}(\mc{L}_\text{alg}(\mc{O}))$ was shown to be an atomistic Boolean algebra, which is isomorphic with the lattice $\mc{L}_\text{atom}(\mc{L}_\text{alg}(\mc{O}))$ of subsets of atoms in $\mc{L}_\text{alg}(\mc{O})$, ordered by inclusion (Propositions \ref{prop:atom2} and \ref{prop:atom3}). Thus, local finite-dimensionality of physical systems leads to a specific type of point-free granularity of spacetime.

Of course, the assumptions we made along the way could be further weakened. Indeed, the logic of the derivation seems general enough to allow for various possible extensions and modifications. For example, one could consider local observable algebras with non-trivial centers. On the other hand, one could abandon Assumption \ref{ass:com} altogether. In the finite-dimensional case Assumption \ref{ass:com} implies that the local observable algebra $\mf{A}$ is generated by any local subfactor $\mf{B}\subset\mf{A}$ and its relative commutant $\mf{B}'_{\mf{A}}$, which is also a local algebra, so that $\mf{A} \cong \mf{B}\otimes\mf{B}'_\mf{A}$. However, this does not usually hold in gauge theory. In the absence of Assumption \ref{ass:com} we can still show that $\mc{L}_\text{alg}(\mc{O})$ atomic, since it is of finite length, but not necessarily atomistic. In this case, $\mc{T}(\mc{O})\cong\mc{L}_{ds}^{**}(\mc{L}_\text{alg}(\mc{O}))$ is still an atomic frame, but not necessarily complemented (i.e., a Boolean algebra).

As explained in the introduction, for physical reasons we expect the finite-dimensionality of local algebras manifest when gravitational effects become relevant. Thus, the (naive) expectation is that the point-free granularity of spacetime associated with the atomicity of topology should become evident in high energy physics only close to the Planck scale. Accordingly, the experimental consequences of the small scale granularity of spacetime topology are expected to be extremely weak at macroscopic scales. The derivation of some more precise experimental consequences of the proposed spacetime granularity would however require the formulation of a concrete model, which would be locally finite-dimensional, and whose infinite-dimensional limit would coincide with some physical QFT model.\footnote{One could, of course, always just straight-forwardly discretize some QFT model on a lattice with some cutoffs in place, but this is not very satisfactory due to the fixed background structure.} This is a highly interesting direction of research, and one that is being pursued by several authors \cite{Parikh05,Singh18,Cao19}, but is unfortunately outside the scope of this particular paper. Nevertheless, it is interesting to us that one can rigorously infer something concrete about the Planck scale structure of spacetime by such rather general assumptions as we have made here.

\section*{Acknowledgments}
I thank Paolo Bertozzini, Philipp H\"ohn and Ted Jacobson for enlightening discussions and helpful comments on an earlier version of the manuscript.

\end{document}